\documentclass[12pt]{article}
\usepackage{amssymb}
\usepackage{eucal}
\usepackage{amsmath}
\usepackage{amsthm}
\usepackage{amssymb}
\usepackage{eucal}
\usepackage{amscd}
\usepackage{graphics}
\usepackage{graphicx}
\usepackage{amsfonts}
\usepackage{latexsym}
\usepackage[all]{xy}
\usepackage{enumerate}
\usepackage{framed}
\usepackage[dvipsnames]{xcolor}

\newcommand{\R}{{\mathbb  R}}

\numberwithin{equation}{section}
\newtheorem{thm}{\bf Theorem}[section]
\newtheorem{lem}[thm]{\bf Lemma}

\newtheorem{defn}{\bf Definition}[section]
 
\theoremstyle{remark}
\newtheorem{rem}{\bf Remark}[section]

\newtheorem{exmp}{\bf Example}[section]

\begin{document}

\title{The Steady States of Antitone Electric Systems}
\author{Dan Com\u anescu\\
{\small Department of Mathematics, West University of Timi\c soara}\\
{\small Bd. V. P\^ arvan, No 4, 300223 Timi\c soara, Rom\^ ania}\\
{\small E-mail addresses: dan.comanescu@e-uvt.ro}}
\date{}

\maketitle

\begin{abstract}
The steady states of an antitone electric system are described by an antitone function with respect to the componentwise order. When this function is bounded from below by a positive vector, it has only one fixed point. This fixed point is attractive for the fixed point iteration method. In the general case, we find existence and uniqueness results of fixed points using the set of $\mathbb{P}$-matrices.  
\end{abstract}


\section{Introduction}

Some relevant mathematical models of power systems have the steady states modeled by the system\footnote{If ${\bf y}:=(y_1,\dots,y_n)^t$, then $\frac{1}{\bf y}=\left(\frac{1}{y_1},\dots, \frac{1}{y_n}\right)^t$.}
\begin{equation}\label{ecuatia-fix-initial}
{\bf y}={\bf k}-\widetilde{M}\frac{1}{\bf y},
\end{equation}
where the unknown ${\bf y}$ is an element of $(\R^*)^n$ and ${\bf k}\in \R^n$, $\widetilde{M}\in \mathcal{M}_n(\R)$ are parameters. In this paper we call the above system the electric power system.

As a first example, we specify that the steady states of a linear time invariant DC system with CPLs, see \cite{polyak} and \cite{comanescu}, are the solutions $({\bf y,u})\in \R^n\times \R^n$ of the problem 
\begin{equation}
\begin{cases}
{\bf y}=\overline{M}{\bf u}+{\bf k} \\
{\bf y}\circ {\bf u}=-{\bf P},
\end{cases}
\end{equation}
where $\overline{M}\in \mathcal{M}_n(\R)$, ${\bf P}\in \R^n$, ${\bf k}\in \R^n$, and ''$\circ $'' is the Hadamard product\footnote{If ${\bf y},{\bf z}\in \R^n$ have the components $y_1,\dots, y_n$ and ${z}_1,\dots, {z}_n$, then the Hadamard product (or the Schur product, \cite{chandler}) ${\bf y}\circ{\bf z}\in \R^n$ has the components $y_1{z}_1,\dots,y_n{z}_n$ (see \cite{horn}).}.
Using the notation $\widetilde{M}:=\overline{M}\text{diag}({\bf P})$ we can reduce this problem to the system \eqref{ecuatia-fix-initial}.  
In the specialized literature the above system is studied making certain assumptions about the matrix $M$. In \cite{sanchez} and \cite{polyak} the symmetric part of $\overline{M}$ is considered to be positive definite.   
In \cite{matveev} it is assumed that $A:=\overline{M}^{-1}$ is a Stieltjes matrix\footnote{The symmetric positive definite matrix $A=\left[A_{ij}\right]\in \mathcal{M}_n(\R)$ is called a Stieltjes matrix if for $i\neq j$ we have $A_{ij}\leq 0$ (see \cite{micchelli}).}. 

A second example comes from the papers \cite{jeeninga-1} and \cite{jeeninga-2} where a DC power grid model with constant-power loads at steady state is considered. 
For a given power grid with $n$ loads and $m$ sources the voltage potentials ${\bf V}$, the power ${\bf P}$, and the Kirchhoff matrix $Y$ are partitioned according to whether nodes are loads ($L$) and sources ($S$)
$${\bf V}=\begin{pmatrix}{\bf V}_L \\ {\bf V}_S\end{pmatrix}\in \R^{n+m}, {\bf P}=\begin{pmatrix}{\bf P}_L \\ {\bf P}_S\end{pmatrix}\in \R^{n+m},\,\,\text{and}\,\,Y=\begin{pmatrix} Y_{LL} & Y_{LS} \\ Y_{SL} & Y_{SS} \end{pmatrix}\in \mathcal{M}_{n+m}(\R).$$
The DC power flow equations for constant-power loads are given by
\begin{equation}\label{Y_LL-eq}
\text{diag}({\bf V}_L)Y_{LL}({\bf V}_L-{\bf V}_L^*)+{\bf P}_c={\bf 0},
\end{equation}
where ${\bf V}_L^*\in (\R_+^*)^n$ is the vector of the open-circuit voltages and ${\bf P}_c\in \R^n$ is the vector of constant power demands at the loads. 
If we note ${\bf y}:={\bf V}_L$, ${\bf k}:={\bf V}_L^*$, and $\widetilde{M}:=Y_{LL}^{-1}\text{diag}({\bf P}_c)$ the above system can be reduced to the system \eqref{ecuatia-fix-initial}. 
All voltage potentials are assumed to be positive. The equation \eqref{Y_LL-eq} is feasible for ${\bf P}_c$ if it has a positive solution ${\bf V}_L$ (operating point associated to ${\bf P}_c$). 
In the papers \cite{jeeninga-1} and \cite{jeeninga-2} is presented a study of the solutions of \eqref{Y_LL-eq} and are obtained necessary and sufficient conditions for feasibility. 

In many situations a component of the unknown vector ${\bf y}$ of the system \eqref{ecuatia-fix-initial} has a fixed sign.  Possibly making a change of variable we can assume that ${\bf y}$ is positive. The positive solutions of \eqref{ecuatia-fix-initial} are the fixed points of the function $F_{{\bf k},\widetilde{M}}:\left(\R_+^*\right)^n\to \R^n$ given by $F_{{\bf k},\widetilde{M}}({\bf y})={\bf k}-\widetilde{M}\frac{1}{\bf y}.$

In the paper \cite{comanescu} we studied the isotone electric system case characterized by the condition that $\widetilde{M}$ is a nonnegative matrix. The solutions of the isotone electric system are the fixed points of an isotone function with respect to the componentwise order\footnote{${\bf x}\leq {\bf y}\,\,\leftrightarrow \,\,x_i\leq y_i \,\,\forall i\in \{1,\dots,n\}$ (see \cite{ehrgott}).} ''$\leq$'' on $\R^n$.

In this paper, we continue the study of electric power systems of the form \eqref{ecuatia-fix-initial}, assuming that $M:=-\widetilde{M}$ is a nonnegative matrix and the unknown vector ${\bf y}$ is positive. The system \eqref{ecuatia-fix-initial} becomes {\bf the antitone electric system}
\begin{equation}\label{ecuatia-fix-123}
{\bf y}={\bf k}+M\frac{1}{\bf y}.
\end{equation}
The positive solutions of \eqref{ecuatia-fix-123} are the fixed points of the function $S_{{\bf k},M}:\left(\R_+^*\right)^n\to \R^n$ given by
\begin{equation}\label{functia-F}
S_{{\bf k},M}({\bf y})={\bf k}+M\frac{1}{\bf y}.
\end{equation}
The function $S_{{\bf k},M}$ is antitone with respect to the componentwise order; for ${\bf y},\overline{\bf y}\in \left(\R_+^*\right)^n$ with ${\bf y}\leq \overline{\bf y}$ we have $S_{{\bf k},M}(\overline{\bf y})\leq S_{{\bf k},M}({\bf y})$. $\Phi_{S_{{\bf k},M}}$ is the set of fixed points of $S_{{\bf k},M}$.

In Section \ref{sec-general-99} we present some properties of the fixed points of a continuous, antitone function $S:\left(\R_+^*\right)^n\to \R^n$. When $S$ is bounded from below by a positive vector, we use our study \cite{comanescu} on the fixed points of a continuous isotone function to prove the existence of a fixed point.  In this case, the set of continuous antitone functions is divided into the set of type I functions and the set of type II functions. We show that a type I function has a single fixed point that is attractive for the fixed point iteration method. In the general case, we find existence and uniqueness results of fixed points using the Jacobian matrix and the set of $\mathbb{P}$-matrices (see Appendix \ref{P-matrices-08}).  

In Section \ref{main-1} we study the steady states of an antitone electric system. We prove that $S_{{\bf k},M}$ is a continuous antitone function of type I. When ${\bf k}$ is positive, $S_{{\bf k},M}$ has a single fixed point which is attractive for the fixed point iteration method. For ${\bf k}\in \R^n$ and $M$ a nonnegative $\mathbb{P}_0$-matrix $S_{{\bf k},M}$ has at most one fixed point. For ${\bf k}\in \R^n$ and $M$ a nonnegative $\mathbb{P}_0$-matrix with positive elements on the diagonal, $S_{{\bf k},M}$ has at least one fixed point.

Some examples which facilitate the understanding of theoretical aspects are studied in Section \ref{examples-22}. 

At the end of the paper we present some notions and results used in the main sections.

\section{The fixed points of an antitone function}\label{sec-general-99}

We present some results about the fixed points of $S:\left(\R_+^*\right)^n\to \R^n$ which is assumed to be a continuous antitone function. We pay a special attention to the case when $S$ is bounded from below.

\subsection{The case $n=1$}

For $n=1$, an elementary study leads to the following results:
\begin{enumerate}[(i)]
\item if $\lim_{y\to 0}S(y)\leq 0$, then $S$ has no fixed points;
\item if $\lim_{y\to 0}S(y)> 0$, then $S$ has a single fixed point $y_S^{\Box}$. 
\end{enumerate}

An important situation is obtained when the inequalities $\lim_{y\to \infty}S(y)\geq k>0$ are verified. In this case the domain of $S\circ S$ is $\R_+^*$ and the sequence $\text{Fit}_S({y}_0)$, generated by the fixed point iteration method\footnote{$\text{Fit}_S({y}_0):=(y_r)_{r\in \mathbb{N}}$ with $y_{r+1}=S(y_r)$, $r\in \mathbb{N}$.}, has an infinity of terms. The sequence $\text{Fit}_{S\circ S}(k)$ is isotone\footnote{If $\text{Fit}_{S\circ S}(k):=(u_r)_{r\in \mathbb{N}}$, $u_0=k$, then $u_{r}\leq u_{r+1}$, $r\in \mathbb{N}$.}, convergent with the limit $\omega_{S\circ S}(k)$.  The sequence $\text{Fit}_{S\circ S}(S(k))$ is antitone\footnote{If $\text{Fit}_{S\circ S}(S(k)):=(v_r)_{r\in \mathbb{N}}$, $v_0=S(k)$, then $v_{r+1}\leq v_{r}$, $r\in \mathbb{N}$.}, convergent with the limit $\omega_{S\circ S}(S(k))$. We observe that we have $\omega_{S\circ S}(k)\leq \omega_{S\circ S}(S(k))$. 

{\bf I.} If $S$ is of type I, $\omega_{S\circ S}(k)=\omega_{S\circ S}(S(k))$, then $\omega_{S\circ S}(k)=\omega_{S\circ S}(S(k))=y_S^{\Box}$ (the fixed point of $S$) and it is globally attractive for the fixed point iteration method (for all $y_0>0$ we have $y_r\to y_S^{\Box}$, $\text{Fit}_{S}(y_0):=(y_r)_{r\in \mathbb{N}}$). 

{\bf II.} If $S$ is of type II, $\omega_{S\circ S}(k)<\omega_{S\circ S}(S(k))$, then $\omega_{S\circ S}(k)<y_S^{\Box}<\omega_{S\circ S}(S(k))$, the $\omega$-limit set of $S$ (see Appendix \ref{fixed-point-iteration-section}) satisfies $\Omega_S\subset[\omega_{S\circ S}(k),\omega_{S\circ S}(S(k))]$, and $\Omega_S$ can be very rich in elements (see Example \ref{exemplu-complex-55}). 

\subsection{$S$ is bounded from below by a positive vector}

In this paragraph we assume that the continuous antitone function $S:\left(\R_+^*\right)^n\to \R^n$ is bounded from below by ${\bf k}\in \left(\R_+^*\right)^n$; for all ${\bf y}\in \left(\R_+^*\right)^n$ we have ${\bf k}\leq S({\bf y})$. 

In this case $S\circ S:\left(\R_+^*\right)^n\to \R^n$ is a continuous isotone function. Also, $S\circ S$ is bounded from above. From Theorem 2.5, \cite{comanescu}, when $\Phi_{S\circ S}\neq \emptyset$, there exists ${\bf y}_{S\circ S}^{\Box}$, the dominant fixed point of $S\circ S$. This special fixed point of $S\circ S$ dominates all $\omega$-limit points of $S\circ S$ (for the fixed point iteration method). 
In what follows we present some properties of the function $S\circ S$.

\begin{lem}\label{ScircS-prop}
In the assumptions of this paragraph the following statements are valid.
\begin{enumerate}[(i)]
\item $S\circ S\left( \left(\R_+^*\right)^n\right)\subset [{\bf k}, S({\bf k})]:=\{{\bf y}\in \R^n\,|\,{\bf k}\leq {\bf y}\leq S({\bf k})\}$.
\item The fixed point iteration sequence $\emph{Fit}_{S\circ S}({\bf k})$, see Appendix \ref{fixed-point-iteration-section}, is isotone and convergent with the limit $\omega_{S\circ S}({\bf k})$. $\emph{Fit}_{S\circ S}(S({\bf k}))$ is antitone and convergent with the limit $\omega_{S\circ S}(S({\bf k}))$.
\item $\Phi_{S\circ S}\neq \emptyset$.
\item $\omega_{S\circ S}({\bf k})=S(\omega_{S\circ S}(S({\bf k})))\leq \omega_{S\circ S}(S({\bf k}))=S(\omega_{S\circ S}({\bf k}))={\bf y}_{S\circ S}^{\Box}$.
\item The $\omega$-limit set of $S\circ S$, see Appendix \ref{fixed-point-iteration-section}, satisfies $\Omega_{S\circ S}\subset [\omega_{S\circ S}({\bf k}), \omega_{S\circ S}(S({\bf k}))]$.
\end{enumerate}
\end{lem}

\begin{proof}
$(i)$ From hypotheses, for ${\bf y}\in \left(\R_+^*\right)^n$, we have ${\bf k}\leq S({\bf y})$ and  ${\bf k}\leq S\circ S({\bf y})$. Using the fact that $S$ is antitone we obtain $S\circ S({\bf y})\leq S({\bf k})$.

$(ii)$ 
Because $S\circ S$ is isotone and ${\bf k}\leq S\circ S({\bf k})$ we obtain, by induction, the fact that  $\text{Fit}_{S\circ S}({\bf k})$ is isotone. It is bounded from above by $S({\bf k})$. Consequently, it is convergent. Analogously, the second statement is proved.

$(iii)$ From $(ii)$ we have that $\Omega_{S\circ S}\cap \left(\R_+^*\right)^n\neq \emptyset$. Using Theorem 2.7, \cite{comanescu}, we obtain that $\Phi_{S\circ S}\neq \emptyset$.

$(iv)$ Using Lemma 2.1-$(iii)$ from  \cite{comanescu} we obtain $\omega_{S\circ S}({\bf k})\leq \omega_{S\circ S}(S({\bf k}))$. We observe that ${\bf y}_{S\circ S}^{\Box} \leq S({\bf k})$. From Theorem 2.5-$(v)$, \cite{comanescu}, we obtain $\omega_{S\circ S}(S({\bf k}))={\bf y}_{S\circ S}^{\Box}$. We make the notation  $\text{Fit}_{S}({\bf k})=({\bf k}_r)_{r\in \mathbb{N}}$ and we have $\text{Fit}_{S\circ S}({\bf k})=({\bf k}_{2r})_{r\in \mathbb{N}}$, $\text{Fit}_{S\circ S}(S({\bf k}))=({\bf k}_{2r+1})_{r\in \mathbb{N}}$. Because ${\bf k}_{2r+1}=S({\bf k}_{2r})$ we obtain $\omega_{S\circ S}(S({\bf k}))=S(\omega_{S\circ S}({\bf k}))$. Because $\omega_{S\circ S}({\bf k})$ is a fixed point of $S\circ S$ we deduce that $S(\omega_{S\circ S}(S({\bf k})))= \omega_{S\circ S}({\bf k})$.

$(v)$ If ${\bf y}_0\in \left(\R_+^*\right)^n$ and $\text{Fit}_{S}({\bf y}_0)=({\bf y}_r)_{r\in \mathbb{N}}$, then
${\bf k}\leq {\bf y}_2=S\circ S({\bf y}_0)\leq S({\bf k})$. Because the set equality $\omega_{S\circ S}({\bf y}_0)=\omega_{S\circ S}({\bf y}_2)$ is satisfied, using Lemma 2.1, $(i)$ and $(ii)$, \cite{comanescu}, we have the announced inclusion.
\end{proof}
The $\Omega$-limit set of $S$ satisfies the following property.
\begin{thm}\label{Omega-S-ScircS}
In the assumptions of this paragraph we have $\Omega_{S}=\Omega_{S\circ S}$.
\end{thm}

\begin{proof}
If ${\bf y}_0\in \left(\R_+^*\right)^n$ and  $\text{Fit}_{S}({\bf y}_0)=({\bf y}_r)_{r\in \mathbb{N}}$, then we have $\text{Fit}_{S\circ S}({\bf y}_0)=({\bf y}_{2r})_{r\in \mathbb{N}}$ and $\text{Fit}_{S\circ S}({\bf y}_1)=({\bf y}_{2r+1})_{r\in \mathbb{N}}$.
We observe that $\omega_{S\circ S}({\bf y}_0)\subset \omega_{S}({\bf y}_0)$ and we obtain $\Omega_{S\circ S}\subset\Omega_{S}$. 

Let be ${\bf y}\in \Omega_{S}$. There exists ${\bf y}_0\in\left(\R_+^*\right)^n$ and
there exists the strictly isotone sequence $(r_q)_{q\in \mathbb{N}}$ of $\mathbb{N}$ such that $({\bf y}_{r_q})_{q\in \mathbb{N}}$, the subsequence of $\text{Fit}_{S}({\bf y}_0)$, is convergent with the limit ${\bf y}$. If $({\bf y}_{r_q})_{q\in \mathbb{N}}\cap 2\mathbb{N}$ has an infinity of elements, then ${\bf y}\in  \omega_{S\circ S}({\bf y}_0)$. If $({\bf y}_{r_q})_{q\in \mathbb{N}}\cap 2\mathbb{N}+1$ has an infinity of elements, then ${\bf y}\in  \omega_{S\circ S}({\bf y}_1)$. We deduce that $ \Omega_{S}\subset \Omega_{S\circ S}$. 

Consequently, $\Omega_{S}= \Omega_{S\circ S}$. 
\end{proof}

The existence of fixed points is ensured for a continuous antitone function that is bounded from below by a positive vector.

\begin{thm}\label{phi=neq-01}
If the continuous antitone function $S:\left(\R_+^*\right)^n\to \R^n$ is bounded from below by ${\bf k}\in \left(\R_+^*\right)^n$, then $\Phi_S\neq \emptyset$ and $\Phi_S\subset [\omega_{S\circ S}({\bf k}), \omega_{S\circ S}(S({\bf k}))]$. 
\end{thm}

\begin{proof}
If $\omega_{S\circ S}({\bf k})\leq {\bf y}\leq  \omega_{S\circ S}(S({\bf k}))$, then we can write $\omega_{S\circ S}({\bf k})=S(\omega_{S\circ S}(S({\bf k})))\leq S({\bf y})\leq S(\omega_{S\circ S}({\bf k}))= \omega_{S\circ S}(S({\bf k}))$. The continous function $S_{|[\omega_{S\circ S}({\bf k}), \omega_{S\circ S}(S({\bf k}))]}$ maps the compact convex set $[\omega_{S\circ S}({\bf k}), \omega_{S\circ S}(S({\bf k}))]$ to itself. From Brouwer fixed point theorem, see Theorem 4.2.5, \cite{istratescu}, we have that  $S_{|[\omega_{S\circ S}({\bf k}), \omega_{S\circ S}(S({\bf k}))]}$ has a fixed point. From Lemma \ref{ScircS-prop} and Theorem \ref{Omega-S-ScircS} we have the announced inclusion.  
\end{proof}

A continuous and antitone function $S$, bounded from below by the positive vector ${\bf k}$, can be of the following two types\footnote{Let be ${\bf y}, \overline{\bf y}\in \R^n$. We have ${\bf y}\lneq \overline{\bf y}$ if and only if ${\bf y}\leq \overline{\bf y}$ and ${\bf y}\neq \overline{\bf y}$.}:
\begin{enumerate}[{\bf type I}]
\item if $\omega_{S\circ S}({\bf k})= \omega_{S\circ S}(S({\bf k}))={\bf y}_{S\circ S}^{\Box}$;
\item if $\omega_{S\circ S}({\bf k})\lneq \omega_{S\circ S}(S({\bf k}))={\bf y}_{S\circ S}^{\Box}$.
\end{enumerate}
In many practical situations, for a specific continuous antitone function $S$, we can numerically decide whether $S$ is of type  I or it is of type II. For ${\bf k}\in \left(\R_+^*\right)^n$ and $M$ a nonnegative matrix, $S_{{\bf k},M}$ of the form \eqref{functia-F} is of type I. The continuous antitone functions used in Example \ref{exemplu-complex-55} and Example \ref{exemplu-complex-2D-09} are of type II.

\subsubsection{Functions of type I.}

\begin{thm}\label{global-attractiv-99}
Suppose that the continuous antitone function $S:\left(\R_+^*\right)^n\to \R^n$ is bounded from below by ${\bf k}\in \left(\R_+^*\right)^n$ and it is of type I. Then,
\begin{enumerate}[(i)]
\item $\Phi_S$ has a single fixed point ${\bf y}_S^{\Box}$ and we have ${\bf y}_S^{\Box}={\bf y}_{S\circ S}^{\Box}$;
\item for all ${\bf y}_0\in \left(\R_+^*\right)^n$ the sequence $\emph{Fit}_{S}({\bf y}_0)$ is convergent and its limit is ${\bf y}_{S}^{\Box}$.
 \end{enumerate}
\end{thm}

\begin{proof}
$(i)$ From Theorem \ref{phi=neq-01} we have $\Phi_{S}=\{{\bf y}_{S\circ S}^{\Box}\}$. 

$(ii)$ If ${\bf y}_0\in \left(\R_+^*\right)^n$, then the sequence $\text{Fit}_{S}({\bf y}_0)$ is bounded and $\omega_S({\bf y}_0)=\{{\bf y}_{S}^{\Box}\}$. Consequently, $\text{Fit}_{S}({\bf y}_0)$ is convergent with the limit ${\bf y}_{S}^{\Box}$.
\end{proof}

From Lemma \ref{ScircS-prop} and Theorem \ref{phi=neq-01} we obtain the following results.

\begin{lem}
Each of the following conditions is a sufficient condition for $S$ to be of type I:
\begin{enumerate}[(i)]
\item $S$ has no comparable elements of order two\footnote{The vector ${\bf y}\in \left(\R_+^*\right)^n$ is an element of order two for $S$ if ${\bf y}\neq S({\bf y})$ and $S\circ S({\bf y})={\bf y}$.}.
\item There is no chain with three elements\footnote{A subset of a partially ordered set is a chain if it is totally ordered with respect to the induced order.} in $\Phi_{S\circ S}$.
\end{enumerate}
\end{lem}

\subsubsection{Functions of type II.}

As we have seen, a 1-D continuous antitone function, bounded below by a strictly positive number, has a sigle fixed point. The continuous antitone function presented in Example \ref{exemplu-complex-2D-09} is of type II and it has an infinity of fixed points.

\subsection{The general case}

It is easy to observe that we have the following result.
\begin{lem}\label{incomparable-789}
If $S:\left(\R_+^*\right)^n\to \R^n$ is an antitone function and ${\bf y},\overline{\bf y}$ are distinct fixed points of the function $S$, then ${\bf y}$ and $\overline{\bf y}$ are incomparable. 
\end{lem}

When the antitone function $S$ is a $\mathcal{C}^1$-function we can use  the Jacobian matrix function $J_S:\left(\R_+^*\right)^n\to \mathcal{M}_n(\R)$ to study the fixed points of $S$. It is easy to observe that for all ${\bf y}\in \left(\R_+^*\right)^n$ the matrix $(-J_{S}({\bf y}))$ is nonnegative. 

We consider the continuous isotone function $\Psi: \left(\R_+^*\right)^n\to \R^n$ given by $\Psi({\bf y})={\bf y}-S({\bf y})$. 
The vector ${\bf y}\in \left(\R_+^*\right)^n$ is a fixed point of $S$ if and only if $\Psi({\bf y})={\bf 0}$.
When $S$ is a $\mathcal{C}^1$-function we have the equality $J_{\Psi}({\bf y})=\mathbb{I}_n-J_S({\bf y})$. 

The set of $\mathbb{P}$-matrices and the set of $\mathbb{P}_0$-matrices, see Appendix \ref{P-matrices-08}, facilitate the demonstration of existence and uniqueness results of fixed points.

\begin{thm}\label{unicity-S-78}
If $S$ is a $\mathcal{C}^1$-function and $\mathbb{I}_n-J_{S}({\bf y})$ is a $\mathbb{P}$-matrix, then $S$ has at most one fixed point. 
\end{thm}

\begin{proof}
In Theorem 4 from \cite{gale} it is shown that {\it a differentiable function $F:[{\bf a},{\bf b}]\subset \R^n\to \R^n$ is injective if ${\bf a}\leq {\bf b}$ and for all ${\bf y}\in [{\bf a},{\bf b}]$ the Jacobian matrix $J_F({\bf y})$ is a $\mathbb{P}$-matrix}. 

For ${\bf 0}<{\bf a}\leq {\bf b}$ and ${\bf y}\in [{\bf a},{\bf b}]$ the Jacobian matrix $J_{\Psi}({\bf y})$ is a $\mathbb{P}$-matrix. By using the above result we have that $\Psi_{|[{\bf a},{\bf b}]}$ is an injective function. 
Because ${\bf a}\leq {\bf b}$ are arbitrary positive vectors we obtain that $\Psi$ is an injective function.
\end{proof}

\begin{rem}
In Example \ref{ex-190}, when $a\in [0.68...,0.83...]$, there exists ${\bf y}\in \left(\R_+^*\right)^n$ such that $\mathbb{I}_n-J_{S_{{\bf k},M}}({\bf y})$ is not a $\mathbb{P}$-matrix and the antitone function $S_{{\bf k},M}$ has several fixed points. In the same Example, for $a\in (0,0.68...)\cup (0.83...,1)$ there exists ${\bf y}\in \left(\R_+^*\right)^n$ such that $\mathbb{I}_n-J_{S_{{\bf k},M}}({\bf y})$ is not a $\mathbb{P}$-matrix and the antitone function $S_{{\bf k},M}$ has a single fixed point.
\end{rem}

In what follows, we present a consequence whose hypotheses are easier to verify in certain situations.

\begin{thm}
If $S$ is a $\mathcal{C}^1$-function and $(-J_{S}({\bf y}))$ is a $\mathbb{P}_0$-matrix, for all ${\bf y}\in \left(\R_+^*\right)^n$, then $S$ has at most one fixed point. 
\end{thm}

\begin{proof}
Lemma 4.8.1, \cite{johnson-smith-tsatsomeros}, states that $A+\varepsilon \mathbb{I}_n$ is a $\mathbb{P}$-matrix when $A$ is a $\mathbb{P}_0$-matrix and $\varepsilon>0$.  In our case, for ${\bf y}\in \left(\R_+^*\right)^n$, the matrix $\mathbb{I}_n-J_{S}({\bf y})$ is a $\mathbb{P}$-matrix. We apply the above theorem. 
\end{proof}

The time has come to present a result about the existence of fixed points.

\begin{thm}\label{existence-general-890}
If $S:\left(\R_+^*\right)^n\to \R^n$ is a $\mathcal{C}^1$ function, it is antitone, for all ${\bf y}\in \left(\R_+^*\right)^n$ the matrix $\mathbb{I}_n-J_S({\bf y})$ is invertible, and for all $r\in \{1,\dots,n\}$ we have\footnote{$S_1,\dots,S_n$ are the components of the vectorial function $S$.} $\lim_{y_r\to 0}S_{r}({\bf y})=\infty$, then $S$ has at least one fixed point. 
\end{thm}

\begin{proof}
We consider the function $F:\left(\R_+^*\right)^n\to \R$ given by\footnote{$\|\cdot\|_2$ is the Euclidean norm on $\R^n$.} $F({\bf y})=\|\Psi({\bf y})\|_2^2$.
By computation we obtain that $\nabla F({\bf y})=2J_{\Psi}({\bf y})^T\Psi({\bf y})$. Using the hypotheses we observe that a stationary point ${\bf y}$ of $F$, $\nabla F({\bf y})={\bf 0}$, verifies $\Psi({\bf y})={\bf 0}$.  

Let be $F^*\in F\left(\left(\R_+^*\right)^n\right)$. We prove the existence of $\boldsymbol{\alpha},\boldsymbol{\beta}\in \left(\R_+^*\right)^n$ with $\boldsymbol{\alpha}<\boldsymbol{\beta}$ such that $F({\bf y})\geq F^*+1$ for ${\bf y}\in \left(\R_+^*\right)^n\backslash [\boldsymbol{\alpha},\boldsymbol{\beta}]$. 

Because $\lim_{y_r\to 0}S_{r}({\bf y})=\infty$ we deduce the existence of $0<\alpha_r$ such that $\Psi_{r}({\bf y})=y_r-S_{r}({\bf y})<-(F^*+1)$ for ${\bf y}$ with $y_r<\alpha_r$. 
We note by $\boldsymbol{\alpha}$ the positive vector which has the components $\alpha_1,\dots,\alpha_n$. For $r\in \{1,\dots,n\}$ we choose $\beta_r>\alpha_r$ such that $\beta_r>S_r(\boldsymbol{\alpha})+F^*+1$. We note by $\boldsymbol{\beta}$ the positive vector which has the components $\beta_1,\dots,\beta_n$.
We have
$\left(\R_+^*\right)^n\backslash [\boldsymbol{\alpha},\boldsymbol{\beta}]=D_1\cup D_2$
with 
$D_1=\{{\bf y}\in\, \left(\R_+^*\right)^n\,|\,\exists r\in \{1,\dots,n\}\,\text{s. t.}\,y_r<\alpha_r\},$ and 
$D_2=\{{\bf y}\in \left(\R_+^*\right)^n\,|\,\boldsymbol{\alpha}\leq {\bf y}\,\text{and}\,\exists r\in \{1,\dots,n\}\,\text{s. t.}\,y_r>\beta_r\}.$

If ${\bf y}\in D_1$, then exists $r\in \{1,\dots,n\}$ such that $y_r<\alpha_r$. In this case it is obtained $\Psi_{r}({\bf y})<-(F^*+1)$ and further we deduce the inequality $F({\bf y})>F^*+1$. 

If ${\bf y}\in D_2$, then $\boldsymbol{\alpha}\leq {\bf y}$ and there exists $r\in \{1,\dots,n\}$ such that $y_r>\beta_r$. We have 
$$\Psi_r({\bf y})=y_r-S_r({\bf y})\geq \beta_r-S_r(\boldsymbol{\alpha})>F^*+1.$$ 
Consequently, $F({\bf y})>F^*+1$.

Because $F$ is a continuous function and $[\boldsymbol{\alpha},\boldsymbol{\beta}]$ is a compact set, there exists ${\bf y}^*\in [\boldsymbol{\alpha},\boldsymbol{\beta}]$ such that $F({\bf y}^*)=\min\{F({\bf y})\,|\,{\bf y}\in [\boldsymbol{\alpha},\boldsymbol{\beta}]\}$. From the continuity of $F$ we have that ${\bf y}^*$ it is not located on the border of $[\boldsymbol{\alpha},\boldsymbol{\beta}]$ and it is a stationary point. We deduce that $\Psi({\bf y}^*)={\bf 0}$ and consequently, ${\bf y}^*\in \Phi_S$.
\end{proof}

Synthesizing the above, we find a result of the existence and uniqueness of fixed points.

\begin{thm}\label{unicity-S-99}
If $S:\left(\R_+^*\right)^n\to \R^n$ is a $\mathcal{C}^1$ function, it is antitone, for all ${\bf y}\in \left(\R_+^*\right)^n$ the matrix $\mathbb{I}_n-J_{S}({\bf y})$ is a $\mathbb{P}$-matrix, and for all $r\in \{1,\dots,n\}$ we have $\lim_{y_r\to 0}S_{r}({\bf y})=\infty$, then $S$ has a single fixed point. 
\end{thm}

\begin{proof}
From Theorem \ref{unicity-S-78} we deduce that $S$ has at most one fixed point. For ${\bf y}\in \left(\R_+^*\right)^n$ the matrix $\mathbb{I}_n-J_{S}({\bf y})$ is invertible because it is a $\mathbb{P}$-matrix. From Theorem \ref{existence-general-890} we obtain the existence of a fixed point for $S$. 
\end{proof}

Fianlly, we present an immediate consequence of the above theorem.

\begin{thm}
If $S:\left(\R_+^*\right)^n\to \R^n$ is a $\mathcal{C}^1$ function, it is antitone, for all ${\bf y}\in \left(\R_+^*\right)^n$ the matrix $(-J_{S}({\bf y}))$ is a $\mathbb{P}_0$-matrix, and for all $r\in \{1,\dots,n\}$ we have $\lim_{y_r\to 0}S_{r}({\bf y})=\infty$, then $S$ has a single fixed point. 
\end{thm}

\section{The steady states of antitone electric systems}\label{main-1}

The steady states of the antitone electric system \eqref{ecuatia-fix-123} are the fixed points of the function $S_{{\bf k},M}$, defined in \eqref{functia-F}, where $M$ is a nonnegative matrix and ${\bf k}\in \R^n$. We observe that the function $S_{{\bf k},M}$ is bounded from below by ${\bf k}$.

\subsection{The antitone electric system for $n=1$}

For $k\in \R$ and $M\in \R_+$ we have $S_{k,M}(y)=k+\frac{M}{y}$.

The function $S_{k,0}$ has fixed points if and only if $k>0$. If $k>0$, then $\Phi_{S_{k,0}}=\{k\}$.

When $M>0$ the set of fixed points is $\Phi_{S_{k,M}}=\left\{\frac{k+\sqrt{k^2+4M}}{2}\right\}$. When $k>0$ we can write $\omega_{S_{k,M}\circ S_{k,M}}(k)=\omega_{S_{k,M}\circ S_{k,M}}(S_{k,M}(k))=\frac{k+\sqrt{k^2+4M}}{2}$ and the sequence $\text{Fit}_{S_{k,M}}(y_0)$ is convergent with the limit $\frac{k+\sqrt{k^2+4M}}{2}$.

\subsection{The case ${\bf k}\in \left(\R_+^*\right)^n$}

We start by presenting an important result for the studied problem.

\begin{lem}\label{dublu-imposibil-9}
If ${\bf k}\in \left(\R_+^*\right)^n$, then there is no vector ${\bf y}\in \left(\R_+^*\right)^n$ of order two\footnote{${\bf y}\in  \left(\R_+^*\right)^n$ is a vector of order two for $S_{{\bf k},M}$ if $S_{{\bf k},M}({\bf y})\neq {\bf y}$ and $S_{{\bf k},M}\circ S_{{\bf k},M}({\bf y})={\bf y}$.} for $S_{{\bf k},M}$ such that\footnote{${\bf y}\lneq {\bf z}$ if ${\bf y}\leq {\bf z}$ and ${\bf y}\neq {\bf z}$.} ${\bf y}\lneq S_{{\bf k},M}({\bf y})$.
\end{lem}

\begin{proof}
Suppose that ${\bf y}\in  \left(\R_+^*\right)^n$ is a vector of order two for $S_{{\bf k},M}$ with ${\bf y}\lneq S_{{\bf k},M}({\bf y})$. If we note by $\overline{\bf y}:=S_{{\bf k},M}({\bf y})$, then we have 
\begin{equation}
{\bf k}+M\frac{1}{\bf y}=\overline{\bf y},\,\,
{\bf k}+M\frac{1}{\overline{\bf y}}={\bf y}.
\end{equation}
We obtain the equality $N(\overline{\bf y}-{\bf y})=\overline{\bf y}-{\bf y}$ with $N:=M\text{diag}\frac{1}{{\bf y}\circ \overline{\bf y}}$. The matrix $N$ is nonnegative and it has 1 as an eigenvalue with the nonnegative vector $\overline{\bf y}-{\bf y}$ as an eigenvector. We deduce that $\mathbb{I}_n-N$ is a singular $\mathcal{Z}$-matrix (see Appendix \ref{P-matrices-08}).
We can write
\begin{align*}
{\bf k}=\overline{\bf y}-M\frac{1}{\bf y}=\overline{\bf y}-N\text{diag}({\bf y}\circ \overline{\bf y})\frac{1}{\bf y}=\overline{\bf y}-N\text{diag}(\overline{\bf y})\text{diag}({\bf y})\frac{1}{\bf y}=(\mathbb{I}_n-N)\overline{\bf y}.
\end{align*}
The $\mathcal{Z}$-matrix $\mathbb{I}_n-N$ and the positive vector $\overline{\bf y}$ satisfy the condition $K_{33}$ from Theorem 1, \cite{plemmons}. We obtain that $\mathbb{I}_n-N$ is a nonsingular $\mathcal{M}$-matrix (see Appendix \ref{P-matrices-08}). Contradiction. 
\end{proof}
Following is presented the main result of this section.

\begin{thm}\label{existenta-unicitate-k-pozitiv}
If ${\bf k}\in \left(\R_+^*\right)^n$, then $\Phi_{S_{{\bf k},M}}$ has a single element ${\bf y}_{S_{{\bf k},M}}^{\Box}$ and for ${\bf y}_0\in \left(\R_+^*\right)^n$ the sequence $\emph{Fit}_{S_{{\bf k},M}}({\bf y}_0)$, given by the fixed point iteration method, is convergent and its limit is ${\bf y}_{S_{{\bf k},M}}^{\Box}$.
\end{thm}

\begin{proof}
From Lemma \ref{ScircS-prop} we have 
$$\omega_{S_{{\bf k},M}\circ S_{{\bf k},M}}({\bf k})=S_{{\bf k},M}(\omega_{S_{{\bf k},M}\circ S_{{\bf k},M}}(S_{{\bf k},M}({\bf k})))\leq \omega_{S_{{\bf k},M}\circ S_{{\bf k},M}}(S_{{\bf k},M}({\bf k}))=S_{{\bf k},M}(\omega_{S_{{\bf k},M}\circ S_{{\bf k},M}}({\bf k})).$$ 
From Lemma \ref{dublu-imposibil-9} we obtain the equality $\omega_{S_{{\bf k},M}\circ S_{{\bf k},M}}({\bf k})=\omega_{S_{{\bf k},M}\circ S_{{\bf k},M}}(S_{{\bf k},M}({\bf k}))$. The announced result is a consequence of Theorem \ref{global-attractiv-99}. 
\end{proof}

\subsection{The general case}

For $M$ a nonnegative matrix and ${\bf k}\in \R^n$ the function $S_{{\bf k},M}$ can have one or more fixed points. Also, there are situations where it has no fixed points (see Example \ref{ex-190}).

First, we present some elementary results about the function $S_{{\bf k},M}$.

\begin{lem}\label{prop-SkM-012}
Let be $M\in \mathcal{M}_n(\R)$ a nonnegative matrix and ${\bf k}\in \R^n$. We note by $S_1,\dots, S_n$ the components of $S_{{\bf k},M}$. The following statements hold. 
\begin{enumerate}[(i)]
\item If $r\in \{1,\dots, n\}$ and $M_{rr}>0$, then $\lim_{y_r\to 0}S_r({\bf y})=\infty$. 
\item The Jacobian function of $S_{{\bf k},M}$ is $J_{S_{{\bf k},M}}({\bf y})=-M\emph{diag}\frac{1}{{\bf y}\circ{\bf y}}$. 
\end{enumerate}
\end{lem}

Complementing the existence and uniqueness result demonstrated in Theorem \ref{existenta-unicitate-k-pozitiv}, for ${\bf k}\in \left(\R_+^*\right)^n$, we present a uniqueness result and an existence result for the general case. 
These results are obtained when $M$ is a nonnegative $\mathbb{P}_0$-matrix.

\begin{thm}\label{existenta-unicitate-electric-11}
Let be $M\in \mathcal{M}_n(\R)$ a nonnegative $\mathbb{P}_0$-matrix and let be ${\bf k}\in \R^n$.
\begin{enumerate}[(i)]
\item $S_{{\bf k},M}$ has at most one fixed point.
\item If all the diagonal entries of $M$ are positive ($M_{ii}>0$, $i\in  \{1,\dots,n\}$), then  $S_{{\bf k},M}$ has a single fixed point.
\end{enumerate}
\end{thm}

\begin{proof}
$(i)$ For ${\bf y}\in \left(\R_+^*\right)^n$ we have $\mathbb{I}_n-J_{S_{{\bf k},M}}({\bf y})=\mathbb{I}_n+M\text{diag}\frac{1}{{\bf y}\circ{\bf y}}$. Because $M$ is a $\mathbb{P}_0$-matrix, using Lemma \ref{caract-P_0}, we obtain that $\mathbb{I}_n-J_{S_{{\bf k},M}}({\bf y})$ is a $\mathbb{P}$-matrix. From Theorem \ref{unicity-S-78} we deduce that $S_{{\bf k},M}$ has at most one fixed point.

$(ii)$ Using Lemma \ref{prop-SkM-012} we have that $\lim_{y_r\to 0}S_r({\bf y})=0$ for all $r\in \{1,\dots, n\}$. The announced result is an immediate consequence of Theorem \ref{unicity-S-99}.
\end{proof}

\begin{rem} In Example \ref{ex-657x}, a particular case is presented with $M$ a nonnegative $\mathbb{P}_0$-matrix that has some zeros on the diagonal. In this case, there are ${\bf k}\in \R^n$ so that the function $S_{{\bf k},M}$ has no fixed points.
\end{rem}

\begin{rem}
From Lemma \ref{caract-P_0} we have that the requirement that the matrix $\mathbb{I}_n-J_{S_{{\bf k},M}}({\bf y})=\mathbb{I}_n+M\text{diag}\frac{1}{{\bf y}\circ{\bf y}}$ be invertible, for all ${\bf y}\in \left(\R_+^*\right)^n$, is equivalent to the requirement that $M$ be a $\mathbb{P}_0$-matrix. In this case, the hypotheses of Theorem \ref{existence-general-890} are equivalent to the hypotheses of Theorem \ref{unicity-S-99}.
\end{rem}

\subsection{Conclusions} 

The antitone electric system \eqref{ecuatia-fix-123} is characterized by a nonnegative matrix $M\in \mathcal{M}_n(\R)$ and a vector ${\bf k}\in \R^n$. The solutions of \eqref{ecuatia-fix-123} are the fixed points of the antitone functions $S_{{\bf k},M}$ given by \eqref{functia-F}. 

When ${\bf k}\in \left(\R_+^*\right)^n$ the function $S_{{\bf k},M}$ has a single fixed point which is attractive for the fixed point iteration method (Theorem \ref{existenta-unicitate-k-pozitiv}).  

For an arbitrary vector ${\bf k}$ the uniqueness is ensured by the condition that $M$ is a nonnengative $\mathbb{P}_0$-matrix. For existence, we add the additional condition that the elements on the diagonal are positive (Theorem \ref{existenta-unicitate-electric-11}). Failure to fulfill the previous requirements may lead to the non-existence of fixed points or to the existence of several fixed points. Two distinct fixed points are incomparable (Lemma \ref{incomparable-789}). 

\section{Examples}\label{examples-22}

In this section, we present some examples to illustrate the previous theoretical results.

\begin{exmp}\label{exemplu-complex-55} Let be the continuous antitone function $S:\R_+^*\to \R$ given by $S(y)=\begin{cases} 3-y & y\in (0,2) \\ 1 & y\in [2,\infty)\end{cases}$. We observe that we have $S\circ S(y)=\begin{cases} 1 & y\in (0,1) \\ y & y\in [1,2] \\ 2 & y\in (2,\infty)\end{cases}$, $\Phi_S=\{1.5\}$, and $\Phi_{S\circ S}=[1,2]$. 

A fixed point iteration sequence has the following $\omega$-limit set:
\begin{enumerate}[(i)]
\item $\omega_S(y_0)=\{1,2\}$, for $y_0\in (0,1]\cup [2,\infty)$;
\item $ \omega_S(y_0)=\{y_0,3-y_0\}$, for $y_0\in (1,2)$. 
\end{enumerate}
The $\omega$-limit sets  of the functions $S$ and $S\circ S$ are $\Omega_S=\Omega_{S\circ S}=[1,2]$.
\end{exmp}

\begin{exmp}\label{exemplu-complex-2D-09} Let be the continuous antitone function $S:\left(\R_+^*\right)^2\to \R^2$ with the components $$S_1(y_1,y_2)=\begin{cases} 3-y_2 & y_2\in (0,2) \\ 1 & y_2\in [2,\infty)\end{cases}\,\,\text{and}\,\, S_2(y_1,y_2)=\begin{cases} 3-y_1 & y_1\in (0,2) \\ 1 & y_1\in [2,\infty)\end{cases}.$$
We obtain the equalities
$$S_1 (S(y_1,y_2))=\begin{cases} 1 & y_1\in (0,1) \\ y_1 & y_1\in [1,2] \\ 2 & y_1\in (2,\infty)\end{cases}\,\,\text{and}\,\,S_2(S(y_1,y_2))=\begin{cases} 1 & y_2\in (0,1) \\ y_2 & y_2\in [1,2] \\ 2 & y_2\in (2,\infty)\end{cases}.$$
The sets of the fixed points of $S$ and $S\circ S$ are 
$$\Phi_S=\{(y_1,y_2)^t\,|\,1\leq y_1,y_2\leq 2\,\text{and}\,y_1+y_2=3\},\,\,\Phi_{S\circ S}=[1,2]\times [1,2].$$
\begin{figure}[h!]
\includegraphics[width=10cm]{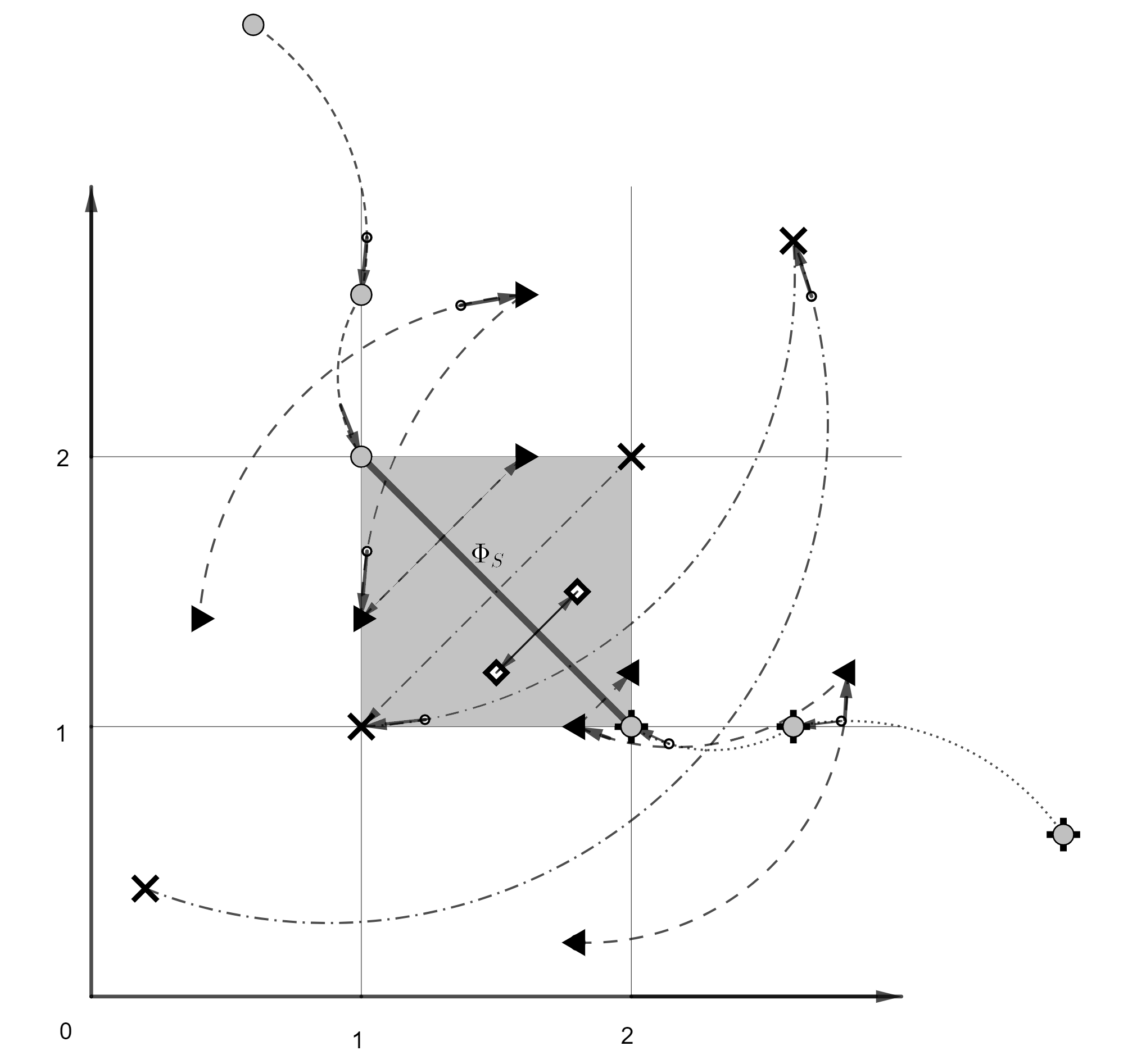}
\centering
\caption{{\small The dynamics generated by the fixed point iteration method for Example \ref{exemplu-complex-2D-09}  (drawing made with GeoGebra).}}\label{ex-general-complex-2D}
\centering
\end{figure}
In Figure \ref{ex-general-complex-2D} is presented the dynamics generated by the fixed point iteration method and the continuous antitone function $S$.
\end{exmp}

\begin{exmp}\label{ex-190}
We consider an antitone electric system which is generated by the function $S_{{\bf k},M}$ with $n=2$, $M=\begin{pmatrix} a & 1 \\ 1 & 1 \end{pmatrix}$, $a\geq 0$. We observe that the symmetric matrix $M$ is a $\mathbb{P}_0$-matrix if and only if $a\geq 1$. Also, $\mathbb{I}_n-J_{S_{{\bf k},M}}({\bf y})$ is a $\mathbb{P}$-matrix, for all ${\bf y}\in \left(\R_+^*\right)^n$, if and only if $a\geq 1$. 

{\bf I: ${\bf k}=(1,1)^t$.} 
$S_{{\bf k},M}$ has a single fixed point which is globally attractive for the fixed point iteration method. Figure \ref{ex-1901-2D} presents the sequence $\text{Fit}_{S_{{\bf k},M}}({\bf k})$ and the fixed point of $S_{{\bf k},M}$ when $a=5$.
\begin{figure}[h!]
\includegraphics[width=8cm]{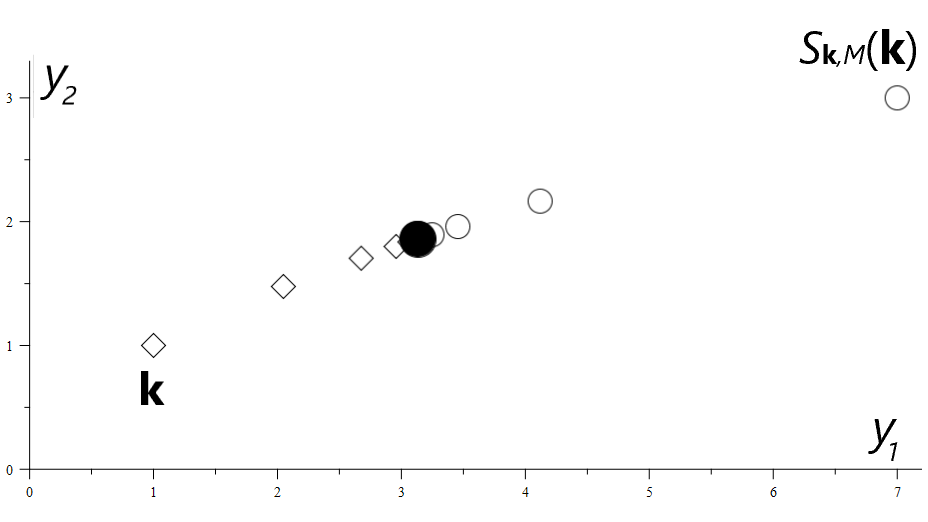}
\centering
\caption{{\small Example \ref{ex-190}/{\bf I}: $\text{Fit}_{S_{{\bf k},M}}({\bf k})$ and the fixed point of $S_{{\bf k},M}$ for $a=5$; $\diamond$ - the terms of $\text{Fit}_{S_{{\bf k},M}\circ S_{{\bf k},M}}({\bf k})$,  $\circ$ - the terms of $\text{Fit}_{S_{{\bf k},M}\circ S_{{\bf k},M}}(S_{{\bf k},M}({\bf k}))$, (drawing made with Maple).}}\label{ex-1901-2D}
\centering
\end{figure}

{\bf II: ${\bf k}=(-9, -10)^t$}.
For $a=0$ the function $S_{{\bf k},M}$ has not fixed points. 

For $a\in (0,0.68...)\cup (0.83...,\infty)$ the function $S_{{\bf k},M}$ has one fixed point. 

For $a\in \{0.68..., 0.83...\}$ the function $S_{{\bf k},M}$ has 2 fixed points.

For $a\in (0.68..., 0.83...)$ the function $S_{{\bf k},M}$ has 3 fixed points.
\end{exmp}

\begin{exmp}\label{ex-657x}
We consider an antitone electric system which is generated the function $S_{{\bf k},M}$ with ${\bf k}=\begin{pmatrix} k \\ 0 \\ 0 \end{pmatrix}\in \R^3$ and $M=\begin{pmatrix} 0 & 0 & 1 \\ 1 & 0 & 1 \\ 0 & 0 & 1 \end{pmatrix}$. We observe that $M$ is a $\mathbb{P}_0$-matrix with some diagonal entries equal to zero.  For $k\leq -1$ the function $S_{{\bf k},M}$ have not fixed points. For $k>-1$ the only fixed point of $S_{{\bf k},M}$ is $(k+1, \frac{k+2}{k+1}, 1)^t$.
\end{exmp}

\appendix

\section{The fixed point iteration sequences, and the $\omega$-limit set}\label{fixed-point-iteration-section}

Let $T:\left(\R_+^*\right)^n\to \left(\R_+^*\right)^n$ be a continuous function. 
The fixed point iteration sequence $\text{Fit}_{T}({\bf y}_0):=({\bf y}_r)_{r\in \mathbb{N}}$ starts from the vector  ${\bf y}_0\in \left(\R_+^*\right)^n$ and it verifies the iteration ${\bf y}_{r+1}=T({\bf y}_r)$, $r\in \mathbb{N}$. 

The $\omega$-limit set of $\text{Fit}_{T}({\bf y}_0)$ is defined by 
\begin{equation}
\omega_T({\bf y}_0)=\{{\bf y}\in \R_+^n\,|\,\exists\,\text{the subsequence}\, ({\bf y}_{k_q})_{q\in \mathbb{N}}\,\text{of}\, \text{Fit}_{T}({\bf y}_0)\, \text{such that} \,{\bf y}_{k_q}\stackrel{q\to \infty}{\longrightarrow} {\bf y}\}.
\end{equation}
From the continuity of $T$ we obtain that the set $\omega_T({\bf y}_0)$ is invariant under $T$. When $\text{Fit}_{T}({\bf y}_0)$ is bounded from above by ${\bf w}$, then it is contained in the compact set $[{\bf 0},{\bf w}]$, and, consequently, $\omega_T({\bf y}_0)\neq \emptyset$. 
When $\text{Fit}_{T}({\bf y}_0)$ is convergent we denote by $\omega_T({\bf y}_0)$ its limit. If $\text{Fit}_{T}({\bf y}_0)$ is convergent and $\omega_T({\bf y}_0)\in \left(\R_+^*\right)^n$, then $\omega_T({\bf y}_0)$ is a fixed point of $T$. 
 
The $\omega$-limit set of $T$ (with respect to the fixed point iteration method) is:
\begin{equation}
\Omega_T=\{{\bf y}\in \R_+^n\,|\,\exists {\bf y}_0\in \left(\R_+^*\right)^n\,\text{such that}\,{\bf y}\in \omega_T({\bf y}_0)\}.
\end{equation}
It is easy to observe that the set of fixed points of $T$ is a subset of $\Omega_T$.

\section{$\mathcal{Z}$-matrix, $\mathcal{M}$-matrix, $\mathbb{P}$-matrix and $\mathbb{P}_0$-matrix}\label{P-matrices-08}

\begin{defn}[see \cite{plemmons}]
Let be $M=[M_{ij}]\in \mathcal{M}_n(\R)$.
\begin{enumerate}[(i)]
\item $M$ is a $\mathcal{Z}$-matrix if for all $i,j\in \{1,\dots,n\}$ with $i\neq j$ we have $M_{ij}\leq 0$.
\item $M$ is an $\mathcal{M}$-matrix if it can be expressed in the form $M = s\mathbb{I}_n- B$, where $B$ is a nonnegative matrix and $s \geq \rho(B)$.
\end{enumerate}
\end{defn}

An $\mathcal{M}$-matrix is a $\mathcal{Z}$-matrix. 

\begin{defn}[see \cite{johnson-smith-tsatsomeros}]
Let be $M\in \mathcal{M}_n(\R)$.
\begin{enumerate}[(i)]
\item $M$ is a $\mathbb{P}_0$-matrix if all its principal minors are nonnegative. 
\item $M$ is a $\mathbb{P}$-matrix if all its principal minors are positive. 
 \end{enumerate}
\end{defn}

We observe that a $\mathbb{P}$-matrix is a $\mathbb{P}_0$-matrix. The paper \cite{johnson-smith-tsatsomeros} presents many properties of the set of $\mathbb{P}$-matrices. A matrix $M\in \mathcal{M}_n(\R)$ with the symmetric part $M^s=\frac{1}{2}(M+M^t)$ a positive definite matrix is a $\mathbb{P}$-matrix.  An invertible $\mathcal{M}$-matrix $M\in \mathcal{M}_n(\R)$ is a $\mathbb{P}$-matrix. If $M$ is an invertible $\mathbb{P}$-matrix, then $M^{-1}$ is a $\mathbb{P}$-matrix. 

The set of $\mathbb{P}_0$-matrices is the topological closure of the set of $\mathbb{P}$-matrices in $\mathcal{M}_n(\R)$. A matrix $M\in \mathcal{M}_n(\R)$ is a $\mathbb{P}_0$-matrix if and only if every real eigenvalue of every principal submatrix of $M$ is nonnegative (Theorem 4.8.2, \cite{johnson-smith-tsatsomeros}).  If $M$ is an invertible $\mathbb{P}_0$-matrix, then $M^{-1}$ is a $\mathbb{P}_0$-matrix. 

To be able to work with $\mathbb{P}$-matrices and $\mathbb{P}_0$-matrices, we use the following notations:
\begin{enumerate}[(i)]
\item $\Sigma_{r}=\{\boldsymbol{\alpha}=(\alpha_1,\dots,\alpha_r)\in \{1,\dots,n\}^r\,|\,\alpha_1<\dots<\alpha_r\}$, $n\in \mathbb{N}^*$ and $r\in \{1,\dots,n\}$.
\item The complement of $\boldsymbol{\alpha}\in \Sigma_r$, $r\in \{1,\dots,n-1\}$, is $\boldsymbol{ \alpha^c}\in \Sigma_{n-r}$ formed with the elements of $\{1,\dots,n\}\backslash\{\alpha_1,\dots,\alpha_r\}$. 
\item $\|\boldsymbol{\alpha}\|=\sum_{s=1}^r\alpha_s$ for $\boldsymbol{\alpha}\in \Sigma_{r}$.
\item ${\bf x}_{\boldsymbol{\alpha}}:=\prod_{s=1}^{r} x_{\alpha_s}$ for ${\bf x}:=(x_1,\dots,x_n)^t\in \R^n$, $r\in \{1,\dots,n\}$, and $\boldsymbol{\alpha}\in \Sigma_{r}$.
\item $A[\boldsymbol{\alpha},\boldsymbol{\beta}]$ is the submatrix of $A\in \mathcal{M}_n(\mathbb{C})$,  
lying in rows $\alpha_1,\dots,\alpha_r$ and columns $\beta_1,\dots,\beta_r$, where $r\in \{1,\dots,n\}$, $\boldsymbol{\alpha}, \boldsymbol{\beta}\in \Sigma_{r}$.
\item  $A(\boldsymbol{\alpha},\boldsymbol{\beta}):=A[\boldsymbol{\alpha}^{\bf c},\boldsymbol{\beta}^{\bf c}]$ is the submatrix obtained from $A$ by deleting the rows $\alpha_1,\dots,\alpha_r$ and columns $\beta_1,\dots,\beta_r$.
\end{enumerate}

For $A,B\in \mathcal{M}_n(\mathbb{C})$ we have the following formula, see \cite{marcus},
\begin{equation}\label{determinant-formula-marcus}
\det(A+B)=\det(A)+\det(B)+\sum_{r=1}^{n-1}\sum_{\boldsymbol{\alpha},\boldsymbol{\beta}\in \Sigma_r}(-1)^{\|\boldsymbol{\alpha}\|+\|\boldsymbol{\beta}\|}\det A[\boldsymbol{\alpha},\boldsymbol{\beta}]\det B(\boldsymbol{\alpha},\boldsymbol{\beta}).
\end{equation}
When $B=\text{diag}({\bf x})\in \mathcal{M}_n(\mathbb{C})$ we observe that  $\det B(\boldsymbol{\alpha},\boldsymbol{\beta})=0$ for $\boldsymbol{\alpha}\neq\boldsymbol{\beta}$ and $\det B(\boldsymbol{\alpha},\boldsymbol{\alpha})={\bf x}_{\boldsymbol{\alpha^c}}:=\prod_{s=1}^{n-r} x_{\alpha^{\bf c}_s}$. 
The formula \eqref{determinant-formula-marcus} becomes
\begin{equation}\label{det-diag-11}
\det(A+\text{diag}({\bf x}))=P_A({\bf x}),
\end{equation}
where  the multivariate polynomial $P_A\in \R[X_1,\dots,X_n]$ is given by 
\begin{equation}\label{polinom-A-675}
P_A({\bf X})=\det(A)+\prod_{s=1}^nX_s+\sum_{r=1}^{n-1}\sum_{\boldsymbol{\alpha}\in \Sigma_r}(\det A[\boldsymbol{\alpha},\boldsymbol{\alpha}]){\bf X}_{\boldsymbol{\alpha^c}}.
\end{equation}

\begin{rem}
In the case when $\text{diag}({\bf x})=-\lambda\mathbb{I}_n$, $\lambda\in \mathbb{C}$, we obtain  the formula of the characteristic polynomial of $A$.
\end{rem}

The following result helps us to determine a characterization of the $\mathbb{P}_0$-matrices.

\begin{lem}\label{lemma-polynomial-11}
Let be $P\in \R[X_1,\dots,X_n]$ of the form $P({\bf X})=X_1\cdot...\cdot X_n+\sum\limits_{r=1}^{n-1}\sum\limits_{\boldsymbol{\alpha}\in \Sigma_r}p_{\boldsymbol{\alpha}}{\bf X}_{\boldsymbol{\alpha}}+p_0$.
The following statements are equivalent:
\begin{enumerate}[(i)]
\item All the coefficients of $P$ are nonnegative real numbers.
\item $P({\bf x})> 0$ for all ${\bf x}\in \left(\R_+^*\right)^n$.
\item $P({\bf x})\geq  0$ for all ${\bf x}\in \R_+^n$. 
\end{enumerate}
\end{lem}

\begin{proof}
$(i)\Rightarrow (ii)$ It is easy to notice that if all the coefficients of $P$ are nonnegative then $P({\bf x})> 0$, ${\bf x}\in \left(\R_+^*\right)^n$. 

By using the continuity of $P$ we obtain $(ii)\Rightarrow (iii)$.

$(iii)\Rightarrow (i)$ Because $P({\bf 0})=p_0$ we obtain $p_0\geq 0$. We have $P(0,\dots,0,x_j,0,\dots,0)=p_jx_j+p_0\geq 0$, $\forall x_j\geq 0$, if and only if $p_j\geq 0$. Consequently for $\boldsymbol{\alpha}\in \Sigma_1$ we have $p_{\boldsymbol{\alpha}}\geq 0$. 

Suppose that $p_{\boldsymbol{\alpha}}\geq 0$, $\boldsymbol{\alpha}\in \Sigma_r$, $r\in \{1,\dots,s\}$. Let be ${\boldsymbol{\alpha}^*}\in \Sigma_{r+1}$. We consider ${\bf x}\in \R_+^n$ with $x_j=0$ for $j\in \{1,\dots,n\}\backslash \{\alpha_1^*,\dots,\alpha_{s+1}^*\}$ and we have 
$P({\bf x})=p_{\boldsymbol{\alpha}^*}{\bf x}_{\boldsymbol{\alpha}^*}+\sum_{r=1}^{s}\sum_{\boldsymbol{\alpha}\in \Sigma_r(\{\alpha_1^*,\dots,\alpha_{s+1}^*\})}p_{\boldsymbol{\alpha}}{\bf x}_{\boldsymbol{\alpha}}+p_0.$
By hypothesis, for $x_{\alpha_1^*}> 0$, ..., $x_{\alpha_{s+1}^*}> 0$,  we have
$ p_{\boldsymbol{\alpha}^*}\geq -\frac{p_0}{{\bf x}_{\boldsymbol{\alpha}^*}}-\sum_{r=1}^{s}\sum_{\boldsymbol{\alpha}\in \Sigma_r(\{\alpha_1^*,\dots,\alpha_{s+1}^*\})}\frac{p_{\boldsymbol{\alpha}}}{{\bf x}_{\boldsymbol{\alpha^c}}}.$
We obtain that $p_{\boldsymbol{\alpha}^*}\geq 0$.
\end{proof}

\begin{lem}\label{caract-P_0}
Let be $A\in \mathcal{M}_n(\R)$. The following statements are equivalent:
\begin{enumerate}[(i)]
\item $A$ is a $\mathbb{P}_0$-matrix.
\item $\det(A+\emph{diag}({\bf x}))>0$ for all ${\bf x}\in \left(\R_+^*\right)^n$.
\item $\det(A+\emph{diag}({\bf x}))\geq 0$ for all ${\bf x}\in \R_+^n$.
\item $\mathbb{I}_n+A\,\emph{diag}({\bf x})$ is invertible for all ${\bf x}\in \left(\R_+^*\right)^n$.
\item $\mathbb{I}_n+A\,\emph{diag}({\bf x})$ is a $\mathbb{P}$-matrix for all ${\bf x}\in \left(\R_+^*\right)^n$.
\end{enumerate}
\end{lem}

\begin{proof}
From Lemma \ref{lemma-polynomial-11} and formulas \eqref{det-diag-11} and \eqref{polinom-A-675} we obtain $(i)\Leftrightarrow (ii) \Leftrightarrow (iii)$.

If ${\bf x}\in \left(\R_+^*\right)^n$ with the components $x_1, \dots, x_n$, then we have the following equality
$$\det(A+\text{diag}({\bf x}))=\det\left(\mathbb{I}_n+A\,\text{diag}\frac{1}{\bf x}\right)\prod_{r=1}^nx_r.$$
It is easy to observe that we have $(ii)\Leftrightarrow (iv)$. 

$(i)\Rightarrow (v)$ If $A$ is a $\mathbb{P}_0$-matrix, then there exists the sequence $(A_r)_{r\in \mathbb{N}}$ of $\mathbb{P}$-matrices such that $A_r\to A$. Using Theorem 4.3.2-(4) from \cite{johnson-smith-tsatsomeros} we deduce that, for ${\bf x}\in \left(\R_+^*\right)^n$, $A_r\text{diag}({\bf x})$ is a $\mathbb{P}$-matrix. Consequently, $A\text{diag}({\bf x})$ is a $\mathbb{P}_0$-matrix. From Lemma 4.8.1, \cite{johnson-smith-tsatsomeros}, we obtain that $\mathbb{I}_n+A\,\text{diag}({\bf x})$ is a $\mathbb{P}$-matrix. 

The implication $(v)\Rightarrow (iv)$ is obvious.
\end{proof}


\begin{thebibliography}{99}
\bibitem{polyak} Barabanov N., Ortega R., Gri\~n\'o R., Polyak B., On Existence and Stability of Equilibria of Linear
Time-Invariant Systems With Constant Power Loads, IEEE Trans. Circuits Syst. I, Reg. Papers, vol. 63, no. 1,
pp. 114–121, Jan. 2016.
\bibitem{chandler} Chandler D., The norm of the Schur product operation, Numerische Mathematik, vol. 4, no. 1, pp. 343–44, 1962. 
\bibitem{comanescu} Com\u anescu D., The steady states of isotone electric systems, Math. Meth. Appl. Sci. (2023), 1–15, DOI 10.1002/mma.9375.
\bibitem{ehrgott} Ehrgott M., Multicriteria Optimization. Second edition, Springer, 2005.
\bibitem{gale} Gale D., Nikaid\^o H., The Jacobian Matrix and Global Univalence of Mappings, Math. Annalen, Vol. 159 (1965), pp. 81-93. 
\bibitem{horn} Horn R.A., Johnson C.R., Matrix Analysis, Second Edition, Cambridge University Press, 2013.
\bibitem{istratescu} Istr\u a\c tescu V.I., Fixed Point Theory. An Introduction, Dordrecht-Boston, Mass.: D. Reidel Publishing Co., 1981.
\bibitem{jeeninga-1} M. Jeeninga, C. De Persis, A. Van der Schaft, {\it DC Power Grids With Constant-Power Loads—Part I: A Full Characterization of Power Flow Feasibility, Long-Term Voltage Stability, and Their Correspondence}, IEEE Trans. Automat. Contr., 68, no. 1 (2022), pp. 2-17.
\bibitem{jeeninga-2} M. Jeeninga, C. De Persis, A. Van der Schaft, {\it DC power grids with constant-power loads—Part II: Nonnegative power demands, conditions for feasibility, and high-voltage solutions}, IEEE Trans Automat. Contr., 68, no. 1 (2022), pp. 18-30.
\bibitem{johnson-smith-tsatsomeros} Johnson C.R., Smith R.L., Tsatsomeros M.J., Matrix positivity, Cambridge Univ. Press, 2020.
\bibitem{marcus} Marcus M., Determinants of Sums, The College Mathematics Journal, vol. 21 (1990), pp. 130-135.
\bibitem{matveev} Matveev A.S., Machado J.E., Ortega R., Schiffer J., Pyrkin A.,  A Tool for Analysis of Existence of Equilibria and Voltage Stability in Power Systems With Constant Power Loads, IEEE Trans. Automat. Contr., Vol. 65, No. 11, Nov. 2020.
\bibitem{micchelli} Micchelli C.A., Willoughby R.A., On Functions Which Preserve the Class of Stieltjes Matrices, Lin. Algebra Appl., Vol. 23 (1979), pp.141-156. 
\bibitem{plemmons} Plemmons R.J., $M$-Matrix Characterizations. I-Nonsingular $M$-Matrices,  Lin. Algebra Appl., Vol. 18 (1977), pp. 175-188. 
\bibitem{sanchez} Sanchez S., Ortega R., Gri\~n\'o R., Bergna G., Molinas-Cabrera M., Conditions for Existence of Equilibrium Points of Systems with Constant Power Loads, IEEE Trans. Circuits Syst. I, Reg. Papers, vol. 61,
no. 7, pp. 2204–2211, July 2014.
\end{thebibliography}
\end{document}